\documentclass[10pt,a4paper]{article}

\usepackage[utf8]{inputenc}
\usepackage[T1]{fontenc}
\usepackage{microtype}
\usepackage{fourier}
\usepackage{tensor}

\usepackage{tikz}
\usepackage[margin=10pt,font=small,labelfont=bf,labelsep=period]{caption}
\usepackage{amssymb,amsthm,amsmath,mathrsfs,bm,braket,marginnote}
\usepackage{mathscinet}
\numberwithin{equation}{section}

\usepackage[nottoc,notlot,notlof]{tocbibind}

\usepackage{aliascnt}
\usepackage[colorlinks,linkcolor=blue,citecolor=blue]{hyperref}

\theoremstyle{plain}

\newtheorem{theorem}{Theorem}[section]

\newaliascnt{lemma}{theorem}
\newtheorem{lemma}[lemma]{Lemma}
\aliascntresetthe{lemma}

\newaliascnt{corollary}{theorem}
\newtheorem{corollary}[corollary]{Corollary}
\aliascntresetthe{corollary}

\theoremstyle{definition}

\newaliascnt{definition}{theorem}
\newtheorem{definition}[definition]{Definition}
\aliascntresetthe{definition}

\newaliascnt{example}{theorem}

\aliascntresetthe{example}

\newaliascnt{remark}{theorem}
\newtheorem{remark}[remark]{Remark}
\aliascntresetthe{remark}

\usepackage{lipsum}

\newcommand{\R}{\mathbf{R}}

\renewcommand{\epsilon}{\varepsilon}

\newcommand{\abs}[1]{\left\lvert #1 \right\rvert}

\DeclareMathOperator{\sgn}{sgn}

\title{Liouville integrability of conservative peakons for a 
modified CH equation\footnote{to appear in J. Nonlinear Math. Phys.}}
\author{Xiangke Chang
\thanks{LSEC, Institute of Computational Mathematics and Scientific Engineering Computing, AMSS, Chinese Academy of Sciences, P.O.Box 2719, Beijing 100190, PR China and Department of Mathematics and Statistics, University of Saskatchewan, 106 Wiggins Road, Saskatoon, Saskatchewan, S7N 5E6, Canada; changxk@lsec.cc.ac.cn}
  \and 
  Jacek Szmigielski
  \thanks{Department of Mathematics and Statistics, University of Saskatchewan, 106 Wiggins Road, Saskatoon, Saskatchewan, S7N 5E6, Canada; szmigiel@math.usask.ca}}
\date{}

\makeatletter
\hypersetup{%
  pdfauthor={Xiangke Chang, Jacek Szmigielski},
  pdftitle={\@title},
}
\makeatother

\begin{document}
\maketitle
\begin{abstract} 
The modified Camassa-Holm equation (also called FORQ) is one of numerous 
\textit{cousins} of the Camassa-Holm equation possessing non-smoth 
solitons (\textit{ peakons}) as special solutions. The peakon 
sector of solutions 
is not uniquely defined: in one peakon sector (dissapative)  the 
Sobolev $H^1$ norm is not preserved, in the other sector (conservative), 
introduced in \cite{chang-szmigielski-short1mCH},   
the time evolution of peakons leaves the $H^1$ norm invariant.
In this Letter, it is shown that the 
conservative peakon equations of the modified Camassa-Holm can 
be given an appropriate Poisson structure relative to which the equations are Hamiltonian and, in fact, Liouville integrable.  The latter is proved directly 
by exploiting the inverse spectral techniques, especially 
asymptotic analysis of solutions, developed elsewhere \cite{chang-szmigielski-m1CHlong}.   
\end{abstract} 

\section{Introduction} 

\vspace{0.5 cm} 

The partial differential equation with cubic nonlinearity 
\begin{equation}\label{eq:m1CH}
m_t+\left((u^2-u_x^2) m)\right)_x=0, \qquad 
m=u-u_{xx}, 
\end{equation}
is a modification of the Camassa-Holm equation (CH) \cite{CH}
\begin{equation} \label{eq:CH} 
m_t+u m_x +2u_x m=0, \, ~m=u-u_{xx}, 
\end{equation} 
for the shallow water waves.  
The history of \eqref{eq:m1CH} is slightly convoluted: it appeared in the papers of Fokas \cite{fokas1995korteweg}, Fuchssteiner \cite{fuchssteiner1996some},  Olver and Rosenau\cite{olver1996tri} and was, later, rediscovered by Qiao \cite{qiao2006new,qiao2007new}.  
Both equations have non-smooth solitons (called \textbf{peakons}) 
as solutions.  
In our recent Letter \cite{chang-szmigielski-short1mCH} we pointed 
out that \eqref{eq:m1CH} has in fact two meaningful types of 
peakon solutions: one type of peakon flows, based 
on the concept of weak solutions to conservation laws, proposed in 
\cite{gui2013wave}, does not preserve the $H^1$ norm $||u||_{H^1}$ , the 
other type put forward in \cite{chang-szmigielski-short1mCH} does.  
For this reason we will refer to these two types of peakons as dissipative, 
conservative, respectively.  

{\bf In this Letter we amplify the message of 
\cite{chang-szmigielski-short1mCH} by showing that in fact the conservative peakons form a Liouville integrable 
Hamiltonian system.} 

  In the remainder of the introduction we argue why this 
is interesting, and certainly not automatic.  Let us briefly recall the 
peakon setup for the CH equation \eqref{eq:CH}, essentially in its original formulation \cite{CH}.  The \textit{peakon Ansatz}
\begin{equation*} 
u=\sum_{j=1}^n p_j(t) e^{-\abs{x-x_j(t)}}
\end{equation*} 
substituted into \eqref{eq:CH} results in the Hamiltonian 
system of equations: 
\begin{equation*} 
\dot x_j=\frac{\partial{H}}{\partial p_j}, \qquad \dot p_j=-
\frac{\partial{H}}{\partial x_j}, 
\end{equation*} 
with the Hamiltonian 
\begin{equation*} 
H=\frac12 \sum_{i,j=1}^n p_i p_j e^{- \abs{x_i-x_j}}. 
\end{equation*}  
Its Liouville integrability was proven in \cite{ragnisco-bruschi} using 
the $R$-matrix formalism.  We empasize that in the CH case the 
amplitudes $p_j$ and positions $x_j$ form conjugate pairs with respect to the 
canonical Poisson structure.

By contrast, in the case of equation \eqref{eq:m1CH}, even though 
the peakon Ansatz looks superficially the same, namely 
$u=\sum_{j=1} m_j(t) e^{-\abs{x-x_j(t)}}$, 
the candidates for momenta $m_j(t)$ are constant and 
one only gets a system of equations on the positions $x_j$ as 
shown in \cite{gui2013wave, qiao-xia} for dissipative peakons and 
in
\cite{chang-szmigielski-short1mCH, chang-szmigielski-m1CHlong} 
for conservative peakons.  
In either case it is not clear from the reduction point of view what 
part of the smooth structure survives the reduction to the peakon sector;  
for 
dissipative peakons, what constitutes the Hamiltonian in the smooth sector, namely the 
square of 
the $H^1$ norm, is not even a constant of motion in the peakon sector even though the norm is well defined.  

To better explain out motivation let us consider  the tempting possibility of 
reaching the conclusions of this paper by transforming 
\eqref{eq:m1CH} to one of the members of the AKNS hierarchy 
by using the reciprocal transformation: 
\begin{equation*} 
x=H(z,t), \qquad dz=mdx-((u_x^2-u^2)m) dt, 
\end{equation*}
discussed in \cite{hone-wang-cubic-nonlinearity, estevez, estevez-sardon} and also \cite{matsuno1mCH} (in \cite{hone-wang-cubic-nonlinearity} 
\eqref{eq:m1CH} is called Qiao's equation).  
The success of such an approach is predicated on finding 
the change of variables $(x,t)\rightarrow (z=F(x,t), \tau=t)$ 
and the inverse $(z,t)\rightarrow (x=H(z,t),t=\tau)$.  
In the smooth sector the existence of $F(x,t)$ is guaranteed if \eqref{eq:m1CH} holds by elementary calculus of 
smooth differential forms (Poincare's Lemma).  In our case, however,  
$m$ is a distribution and \eqref{eq:m1CH} would have to be 
taken as a distributional equation; what remains unclear is which one 
as there is more than one distributional analog of \eqref{eq:m1CH}. 
Either way if $F$ existed it would have to be piecewise constant in $x$ 
since $F_x=m$  and $m$ is a discrete measure, making the transformation 
hard to interpret as a coordinate transformation.  
The situation with the existence of $H(z,t)$ is equally doubtful; 
were $H(z,t)$ to exist it would have to satisfy $H_z=\frac{1}{m}$, the 
inverse of the sum of the Dirac deltas does not seem to have a natural definition.   This state of affairs is not without precedent both in 
the physics and mathematical literature.   The case in point, in fact 
very pertinent to this discussion, is 
the 1-D Schr\"odinger equation
\begin{equation*}
-y_{xx} +uy=E y
\end{equation*}
and the \textit{string} equation: 
\begin{equation*}
-v_{\xi \xi}=Em v. 
\end{equation*}
The transformation (see \cite{CHilbv1}) $y=\sqrt[4]{m} v, \, \, \xi=\int_0^x \sqrt{m(\xi ')} d\xi '$ 
makes both equations equivalent, with $u$ and $m$ related by
\begin{equation*} 
u=\frac{(\sqrt[4]{m})_{xx}}{m}, 
\end{equation*}
provided $m$ is $C^2$ and $m>0$.  When $m$ is a discrete measure, as 
it is in our case, the equivalence fails and these two 
equations are no longer equivalent, either physically, or mathematically.

In summary, we find it more compelling to study the 
peakon sector of \autoref{eq:m1CH} directly using 
well developed theory of distributions and then, if warranted, 
to investigate how, and if, to perform singular limits from the 
smooth sector of \eqref{eq:m1CH} (see interesting comments about 
this procedure in \cite{matsuno1mCH}).   

For conservative peakons studied in this Letter, one 
is tempted to expect Liouville integrability based on the fact 
that the conservative peakon equations are derived from the compatibility conditions 
for a certain distributional Lax pair which was 
constructed in \cite{chang-szmigielski-m1CHlong}.  Indeed, we prove directly 
Liouville integrability by taking advantage of the 
inverse spectral solution formulas obtained in \cite{chang-szmigielski-m1CHlong}.

\section{Conservative peakons}
\autoref{eq:m1CH} 
reduces to the conservative peakon sector 
\cite{chang-szmigielski-short1mCH} defined by the Ansatz 
\begin{equation} \label{eq:peakonansatz} 
    u=\sum_{j=1}^n m_j (t)e^{-\abs{x-x_j(t)}}, 
\end{equation} 
and the multiplication rule 
\begin{equation} \label{eq:defprod} 
u_x^2 m\stackrel{def}{=}\langle u_x^2 \rangle m, 
\end{equation}
where $\langle u_x^2 \rangle m$ means that at a point $x_j$ 
of the singular support of $m=2\sum_{j=1}^n m_j \delta_{x_j} $ one multiplies by the arithmetic average 
of the right and left limits of $u_x^2$ at $x_j$.  The ensuing reduction 
is captured  by the  system of ODEs
\begin{equation}\label{eq:xmODE}
\dot m_j=0, \qquad \dot x_j=u^2(x_j)-\langle u_x^2 \rangle(x_j), 
\end{equation}
on the weights $m_j$ of the measure $m$ and the points of singular 
support $x_j$.  For later use we give two more explicit versions of 
of the nontrivial part of the equations of motion: 

\begin{align}\label{mCH_ode}
\dot{x}_j&=
2\sum_{\substack{1\leq k\leq n,\\k\neq j}}m_jm_ke^{-|x_j-x_k|}+\notag 
\\&\sum_{k\neq j, i\neq j}m_im_k(1-\sgn(x_j-x_k)\sgn(x_j-x_i))e^{-|x_j-x_k|-|x_j-x_i|},
\end{align}
and its more succinct form 
\begin{equation}\label{eq:peakonODEs}
\dot{x}_j=2\sum_{\substack{1\leq k\leq n,\\k\neq j}}m_jm_ke^{-|x_j-x_k|}+4\sum_{1\leq i<j<k\leq n}m_im_ke^{-|x_i-x_k|}, \qquad 1\leq j\leq n, 
\end{equation}
valid when $x_1<x_2<\cdots<x_n$.

Inspired by papers by Hone and Wang \cite{hone-wang-prolongation-algebras, hone-wang-cubic-nonlinearity} , we have 
\begin{theorem}\label{thm:hamiltonianODEs} 
The equations \eqref{mCH_ode} for the motion of n peakons in the PDE \eqref{eq:m1CH} (with the condition \eqref{eq:defprod} in place)  are a Hamiltonian vector field:
\begin{align}
\dot x_j=\{x_j,h\},\qquad \dot m_j=\{m_j,h\}, 
\end{align}
for the Hamiltonian 
$$h=\sum_{i,k=1}^nm_im_ke^{-|x_i-x_k|}=\int u(\xi) m d\xi=||u||^2_{H^1}.$$
Here the Poisson bracket is given by
\begin{align}\label{eq:Poisson bracket} 
\{x_i,x_k\}=\sgn(x_i-x_k), \qquad \{m_i,m_k\}=\{m_i,x_k\}=0.
\end{align}
\end{theorem}
\begin{proof}
It is obvious that 
$$\{m_j,h\}=0$$ 
under the above Poisson bracket, which leads to
$$\dot m_j=0.$$
We proceed with the computation of $ \{x_j,h\}$:
\begin{align*}
 \{x_j,h\}&=\left\{x_j,\sum_{i,k=1}^nm_im_ke^{-|x_i-x_k|}\right\}=\left\{x_j,\sum_{i=1}^nm_i^2+2\sum_{1\leq i<k\leq n}m_im_ke^{-|x_i-x_k|}\right\}\\
 &=\left\{x_j,2\sum_{1\leq i<k\leq n}m_im_ke^{-|x_i-x_k|}\right\}=2\sum_{1\leq i<k\leq n}m_im_k\left\{x_j,e^{-|x_i-x_k|}\right\}\\
 &=2\sum_{1\leq i<k\leq n}m_im_k \sgn(x_k-x_i) e^{-|x_i-x_k|}\left(\sgn(x_j-x_i)-\sgn(x_j-x_k)\right)\\
 &=2\sum_{1\leq j=i<k\leq n}+2\sum_{1\leq i<k=j\leq n}+2\sum_{1\leq j<i<k\leq n}+2\sum_{1\leq i<k<j\leq n}+2\sum_{1\leq i<j<k\leq n}, 
\end{align*}
where we suppressed displaying the actual terms in the summation, 
concentrating on the ranges of summation instead.  
The first two summations give
$$
2\sum_{1\leq j=i<k\leq n}+2\sum_{1\leq i<k=j\leq n}=2\sum_{\substack{1\leq k\leq n,\\k\neq j}}m_jm_ke^{-|x_j-x_k|}.
$$
As for the last three three summations, we have
\begin{align*}
& 2\sum_{1\leq j<i<k\leq n}+2\sum_{1\leq i<k<j\leq n}+2\sum_{1\leq i<j<k\leq n}\\
 &=\sum_{k\neq j, i\neq j}m_im_k \sgn(x_k-x_i) e^{-|x_i-x_k|}\left(\sgn(x_j-x_i)-\sgn(x_j-x_k)\right)\\
 &=\sum_{k\neq j, i\neq j}m_im_k(1-\sgn(x_j-x_k)\sgn(x_j-x_i))e^{-|x_j-x_k|-|x_j-x_i|},
 \end{align*}
 where the last equality is based on the facts that the corresponding equality holds for all the following cases 
\begin{align*}
 &x_i=x_k,\quad x_i=x_j,\quad x_j=x_k,\\
  &x_j<x_i<x_k,\quad x_i<x_j<x_k,\quad x_i<x_k<x_j, \\
  &x_j<x_k<x_i,\quad x_k<x_j<x_i,\quad x_k<x_i<x_j.
  \end{align*}
  Thus, we eventually get
  $$
   \{x_j,h\}=2\sum_{\substack{1\leq k\leq n,\\k\neq j}}m_jm_ke^{-|x_j-x_k|}+\sum_{k\neq j, i\neq j}m_im_k(1-\sgn(x_j-x_k)\sgn(x_j-x_i))e^{-|x_j-x_k|-|x_j-x_i|},
  $$
  which reproduces \eqref{mCH_ode} and thus the proof is completed.
 
\end{proof}

\begin{remark} 
Observe that the Poisson bracket used above is 
a limiting case of a family of Poisson brackets 
discussed in \cite{hone-wang-cubic-nonlinearity}.  In our 
case, if we restrict our considerations to the 
space of positions,  then the Poisson bracket $\sgn(x_i-x_k)$ is up to a scale 
the skew-symmetric Green's function of the operator $D_x$.  
\end{remark} 
For the remainder of this Letter we 
will focus on the following Poisson manifold $(M, \pi)$ defined 
with the help of the Poisson bracket \eqref{eq:Poisson bracket}.  

\begin{definition} \label{def:pmanifold} 
Let 
\begin{equation} 
M=\big\{x_1<x_2<\cdots<x_n\big\}
\end{equation}
 and 
 \begin{equation} 
 \pi(f,g)=\{f,g\}=\sum_{1\leq i<j\leq n} \{x_i, x_j\} \frac{\partial{f}}{\partial{x_i}}\frac{\partial{g}}{\partial{x_j}}
\end{equation} 
 be defined for real valued functions $f,g$ on $M$.  The Poisson 
 manifold $M$ with the Poisson structure $\pi$ will be denoted 
 $(M, \pi)$.  
\end{definition}

We will record the following result which trivially follows from 
\eqref{eq:Poisson bracket} and the definition of rank of $\pi$ (see 
e.g. \cite{poisson}).  
\begin{lemma} \label{lem:rank pi}
Let $n=2K$ or $n=2K+1$.  Then 
the $rank(\pi)=2K$.  
\end{lemma} 
\section{A manifold of conservative global peakons} 
To simplify the presentation we will focus mostly on the case $n=2K$. We will use $\mathbf{m}$ throughout  this paper to denote the $n$-tuple of constant masses 
$m_j$.   
First we make a general observation about the nature of the 
vector field in \eqref{eq:xmODE}: the vector field is discontinuous on the hyperplanes $x_j=x_k, j\neq k$.    
Let us then denote by 
$\mathcal{P}$ the set $\{\mathbf{m}; x_1<x_2<\dots<x_{2K}\}$ of masses and positions where 
the vector field (see \eqref{mCH_ode}) is Lipschitz, in fact smooth.  
The scattering map $\mathcal{S}$ used in \cite{chang-szmigielski-m1CHlong} maps $\mathcal{P}$ to the set of admissible 
scattering data $\mathcal{R} =\{\mathbf{m}; d\mu, c\geq 0\}$ consisting of the spectral measure 
$d\mu=\sum_{k=1}^K b_k \delta_{\zeta_k}, \, b_k>0$ and a constant 
$c$ which is $0$ if $n=2K$ and strictly positive if $n=2K+1$.  
The problem is isospectral with the evolution of the spectral measure 
given by $d\mu(t)=e^{\frac{2t}{\zeta} } d\mu(0)$.  
For an arbitrary choice of masses $m_j>0$ and initial positions $x_j(0)$ in 
\eqref{eq:peakonansatz} the flows are in general not global, since  it can happen that $x_j(t)=x_{j+1}(t)$ 
for some $j$ at some finite time $t>0$.  However, there exists a family of 
open subsets of $\mathcal{R}$ (see the Theorem below)  for which the peakon flows are 
global, i.e.  the solutions $x_j(t)$ exist for all $t\in \R$.  This is crucial to our 
argument as we use the asymptotic behaviour of solutions to simplify 
the computations of Poisson brackets.

\begin{theorem}\label{thm:global}[\cite{chang-szmigielski-m1CHlong}, Case $n=2K$]
  Given arbitrary spectral data $$\{b_j>0, \, 0<\zeta_1<\zeta_2<\cdots< \zeta_K: 1\leq j \leq K \}, $$  and denoting 
  by $i'=n+1 -i$, suppose the masses $m_k$ satisfy
\begin{subequations}
\begin{align}
&&\frac{\zeta_K^{\frac{k-1}{2}}}{\zeta_1^{\frac{k+1}{2}}}< m_{(k+1)' }m_{k'}, \quad \text { for all odd }k, \qquad   1\leq k\leq 2K-1, \\
&&\frac{m_{(k+2)'}m_{(k+1)'}}{(1+m_{(k+1)'}^2\zeta_1)(1+m_{(k+2)'}^2\zeta_1)}<
\frac{\zeta_1^{\frac{k+1}{2}}}{\zeta_K^{\frac{k-1}{2}}}
\frac{2\, \textnormal{min}_j(\zeta_{j+1}-\zeta_j)^{k-1}}{(k+1)(\zeta_K-\zeta_1)^{k+1}}, \notag \\
&&\text{ for all odd } k, \qquad 1\leq k\leq 2K-3. \label{eq:seccond}
\end{align}
\end{subequations}
Then the positions obtained from inverse formulas in \cite{chang-szmigielski-m1CHlong} are ordered $x_1<x_2<\cdots<x_{2K}$ and the multipeakon solutions \eqref{eq:peakonansatz} exist for arbitrary $t\in \R$.
\begin{remark} The odd case of $n=2K+1$ is similar, although 
it requires a special care since in addition to $K$ eigenvalues 
$\zeta_j$ we also have an additional constant of motion, called $c$ in  \cite{chang-szmigielski-m1CHlong}, which intuitively plays a role of 
an additional eigenvalue.  
\end{remark}  

\end{theorem}
Finally, for computations, we will need the asymptotic form 
of global solutions.  The theorem below is a slightly abbreviated form 
of the theorem presented in \cite{chang-szmigielski-m1CHlong} .  
\begin{theorem}\label{thm:evenass} Suppose the masses $m_j$ satisfy the conditions of Theorem \ref{thm:global}.  Then the asymptotic position of a $k$-th (counting from the right) peakon as $t\rightarrow+\infty$ is given by 
\begin{subequations}
\begin{align}
&x_{k'}=\frac{2t}{\zeta_{\frac{k+1}{2}}}+
C_k+\mathcal{O}(
e^{-\alpha_k t}), &\textrm{ for some positive } \alpha_k\, , C_k\in \R, \textrm{ and odd } k, \\
&x_{k'}=\frac{2t}{\zeta_{\frac{k}{2}}}+
C_k+\mathcal{O}(
e^{-\alpha_k t}), &\textrm{ for some positive } \alpha_k\, , C_k\in \R \textrm{ and even } k,\\
&x_{k'}-x_{(k+1)'}=\ln m_{(k+1)'}m_{k'} \zeta_{\frac{k+1}{2}}
+\mathcal{O}(e^{-\alpha_k t}), &\textrm{ for some positive } \alpha_k\, \textrm{ and odd } k.   
\end{align}
\end{subequations}

  Likewise, as $t\rightarrow-\infty$, for convenience using the notation 
  $l^*=K+1-l$, the asymptotic position of the $k$-th peakon is given by
  \begin{subequations}
  \begin{align}
&x_{k'}=\frac{2t}{\zeta_{(\frac{k+1}{2})^*}}+
D_k+\mathcal{O}(
e^{\beta_k t}), \qquad \textrm{ for some positive } \beta_k\, , D_k\in \R \textrm{ and odd } k, \\
&x_{k'}=\frac{2t}{\zeta_{(\frac{k}{2})^*}}+
D_k+\mathcal{O}(
e^{\beta_k t}), \qquad \textrm{ for some positive } \beta_k\,, D_k\in \R \textrm{ and even } k,\\
&x_{k'}-x_{(k+1)'}=\ln m_{(k+1)'}m_{k'} \zeta_{(\frac{k+1}{2})^*}
+\mathcal{O}(e^{\beta_k t}), \qquad \textrm{ for some positive } \beta_k\, \textrm{ and odd } k.   
\end{align}
\end{subequations}
\end{theorem}  
\begin{remark} \label{rem:pairs} We emphasize that in spite of notational 
complexity the most  essential features of the theorem can be stated 
simply: asymptotically, particles form bound states consisting of adjacent 
particles sharing asymptotic velocities $\frac{2}{\zeta_j}, j=1, \dots, K$ and 
moving in pairs, each pair moving as if it were a free particle.  
This picture persists if $n=2K+1$ with the only exception that 
for large positive times the first particle, counting from the left, comes to a stop, 
while the remaining particles pair up the same way they do for $n=2K$. 
Likewise, for large negative times, the first particle, counting from the right, 
comes to a stop, while the rest of particles pair up.  

\end{remark} 
We end this subsection with the corollary which will be used 
in the proof of Liouville integrability of the peakon system \eqref{eq:peakonODEs}.  We will state this lemma 
in a slightly informal way by emphasizing the role of asymptotic pairs 
discussed in \autoref{rem:pairs}.  
\begin{corollary} 
Let $n=2K$ then asymptotically pairs of peakons scatter, that is 
the distances between particles from distinct pairs diverge to $\infty$.  
If $n=2K+1$ and one counts the first particle, counting from the left, as a 
`` pair'' then asymptotically, as $t\rightarrow \infty$, the pairs scatter.  Likewise, if one counts the first particle, counting from the right, as a 
`` pair '' then asymptotically, as $t\rightarrow - \infty$, the pairs scatter.  
\end{corollary} 

\subsection{Liouville integrability} 
We need to introduce a bit of notation to 
facilitate the presentation of formulas and subsequent computations.  
Most of the computations in this Letter  involve a choice of $j$-element index sets $I$ and~$J$
from the set $[k] = \{ 1,2,\dots,k \}$.
We will use the notation
$\binom{[k]}{j}$ for the set of all $j$-element subsets of $[k]$, listed in increasing order; for example $I\in \binom{[k]}{j}$ means that 
$I=\{i_1, i_2,\dots, i_j\}$ for some increasing sequence $i_1 < i_2 < \dots < i_j\leq ~k$. 
Furthermore, given the multi-index $I$ we will abbreviate $g_I=g_{i_1}g_{i_2}\dots g_{i_j}$ etc.
 
\begin{definition}\label{def:bigIndi} Let $I,J \in \binom{[k]}{j}$, or $I\in \binom{[k]}{j+1},J \in \binom{[k]}{j}$.  
\mbox{}

Then  $I, J$ are said to be \emph{interlacing} if 
\begin{equation*}
  \label{eq:interlacing}
    i_{1} <j_{1} < i_{2} < j_{2} < \dotsb < i_{j} <j_{j}
\end{equation*}
or, 
\begin{equation*}
    i_{1} <j_{1} < i_{2} < j_{2} < \dotsb < i_{j} <j_{j}<i_{j+1}, 
\end{equation*}
in the latter case.  
We abbreviate this condition as $I < J$ in either case, and, furthermore, 
use the same notation, that is $I<J$,   for $I\in \binom{[k]}{1}, J \in \binom{[k]}{0}$.
\end{definition} 

\begin{enumerate}
\item \textbf{Case $n=2K$.}

It was shown in \cite{chang-szmigielski-m1CHlong} that the quantities 
\begin{equation} \label{eq:Hjs-even}
H_j=\sum_{\substack{I,J \in \binom{[2K]}{j}\\ I<J}} h_I g_J, \qquad 1\leq j\leq K, 
\end{equation}
with $ h_i=m_ie^{x_i},\ \  g_i=m_ie^{-x_i} $, form a set of $K$  constants  of motion for the system \eqref{eq:peakonODEs} in the even case $n=2K$.  In particular the Hamiltonian $h$ in \autoref{thm:hamiltonianODEs} satisfies 
\[
h=2H_1+\sum_{i=1}^{n}m_i^2.
\]

We need the following computational result.  
\begin{lemma} \label{lem:lem1} 
Consider the Poisson bracket given in \autoref{def:pmanifold}.  
Then  
$$\left\{\prod_{p\in I, |I|=i}m_{2p-1}m_{2p}e^{x_{2p-1}-x_{2p}},\prod_{q\in J,|J|=j}m_{2q-1}m_{2q}e^{x_{2q-1}-x_{2q}}\right\}=0.$$
\end{lemma}
\begin{proof}
Let us denote $$\prod_{p\in I, |I|=i}m_{2p-1}m_{2p}e^{x_{2p-1}-x_{2p}}=F, 
\quad  \prod_{q\in J,|J|=j}m_{2q-1}m_{2q}e^{x_{2q-1}-x_{2q}}=G. $$  

Then, by elementary properties of exponentials and basic properties of Poisson brackets, we have 
\begin{equation*}
\left\{F,G\right\}
=FG \sum_{p\in I,q\in J}\left\{x_{2p-1}-x_{2p},x_{2q-1}-x_{2q}\right\}.
\end{equation*}
However,
\begin{align*}
\{x_{2p-1}-x_{2p},x_{2q-1}-x_{2q}\}&=\{x_{2p-1},x_{2q-1}\}-\{x_{2p-1},x_{2q}\}-\{x_{2p},x_{2q-1}\}+\{x_{2p},x_{2q}\}\\
&=\sgn(p-q)-\sgn(p-q)-\sgn(p-q)+\sgn(p-q)=0.
\end{align*}
Thus the claim is proved.
\end{proof}

\begin{theorem} \label{thm:even commute} 
The Hamiltonians  $H_1,\cdots,H_K$  Poisson commute.
\end{theorem}
\begin{proof}
The idea of the proof goes back at least to the work of J. Moser on 
the finite Toda lattice \cite{moser-toda}.  In the nutshell it amounts to the 
following observation: the Poisson bracket $\{H_i, H_j\}(\textbf{x}^0)$ of two 
conserved quantities $H_i, H_j$ is a also conserved and thus 
instead of the fixed point $\textbf{x}^0$ it can be evaluated on the orbit $\textbf{x}(t)$ going through 
$\textbf{x}^0$, in particular in the asymptotic region $t\rightarrow \pm \infty$, if such is accessible.  In our case both asymptotic regions can be used 
but for the sake of argument we will focus on $t\rightarrow -\infty$.  

Note that, in view of \eqref{eq:Hjs-even}, and for large negative times, the only contributing terms  from individual $H_i$s are 
\begin{equation} \label{eq:Hi-asympt} 
\sum_{I \in \binom{[K]}{i}}\prod_{p\in I}m_{2p-1}m_{2p}e^{x_{2p-1}-x_{2p}}, 
\end{equation}  
hence 
\begin{align*}
&\{H_i,H_j\}(\textbf{x}^0)=\{H_i,H_j\}(\textbf{x}(t))=\lim_{t\rightarrow -\infty}\{H_i,H_j\}(\textbf{x}(t))\\
=&\lim_{t\rightarrow -\infty}\left\{\sum_{I \in \binom{[K]}{i}}\prod_{p\in I}m_{2p-1}m_{2p}e^{x_{2p-1}-x_{2p}},\sum_{J \in \binom{[K]}{j}}\prod_{q\in J}m_{2q-1}m_{2q}e^{x_{2q-1}-x_{2q}}\right\}(\textbf{x}(t))\\
=&\lim_{t\rightarrow -\infty}\sum_{I \in \binom{[K]}{i}}\sum_{J \in \binom{[K]}{j}}\left\{\prod_{p\in I}m_{2p-1}m_{2p}e^{x_{2p-1}-x_{2p}},\prod_{q\in J}m_{2q-1}m_{2q}e^{x_{2q-1}-x_{2q}}\right\}(\textbf{x}(t)).
\end{align*}

By using \autoref{lem:lem1}  the final conclusion follows.
\end{proof}

\item
\textbf{Case $n=2K+1$.}

Again, following \cite{chang-szmigielski-m1CHlong}, 
\begin{equation} \label{eq:Hjs-odd} 
H_j=\sum_{\substack{I,J \in \binom{[2K+1]}{j}\\ I<J}} h_I g_J, \qquad 1\leq j\leq K,
\end{equation} 
with $h_i=m_ie^{x_i},\ \ \ g_i=m_ie^{-x_i}$, 
are constants  of motion for the system \eqref{eq:peakonODEs} in the odd case. 

In the odd case, there is an extra constant of motion, 
which can be computed from the value of the Weyl function at $\infty$ in the spectral variable.  The computation is routine and produces (see 
Section III of \cite{chang-szmigielski-m1CHlong} for details regarding the 
Weyl function)
$$c=\frac{\sum\limits_{\substack{I\in \binom{[2K+1]}{K+1},J \in \binom{[2K+1]}{K}\\ I<J}} h_I g_J}{\sum\limits_{\substack{I,J \in \binom{[2K+1]}{K}\\ I<J}} h_I g_J}=\frac{\sum\limits_{\substack{I\in \binom{[2K+1]}{K+1},J \in \binom{[2K+1]}{K}\\ I<J}} h_I g_J}{H_K},$$
which, in turn, gives an extra constant of motion
$$H_c=\sum\limits_{\substack{I\in \binom{[2K+1]}{K+1},J \in \binom{[2K+1]}{K}\\ I<J}} h_I g_J=\prod_{j=1}^{2K+1}m_je^{(-1)^{j+1}x_j},$$
so that $\{H_1,H_2,\cdots, H_k,H_c\}$ form a set of $K+1$  constants  of motion for the system \eqref{eq:peakonODEs} in this case.

\begin{theorem}\label{thm:odd commute}
The Hamiltonians $H_1,\cdots,H_K,H_c$ Poisson commute .
\end{theorem}
\begin{proof}
Again, asymptotically for large negative times, the $H_j$s simplify to 
$$\sum_{I \in \binom{[K]}{j}}\prod_{j\in I}m_{2j-1}m_{2j}e^{x_{2j-1}-x_{2j}}.$$
Then the argument from the previous theorem applies \textit{verbatim},  proving that 
$$\{H_i,H_j\}=0$$
holds.

We turn now to proving that the $H_j$s and the $H_c$ are in involution. 
A direct computation gives: 

\begin{align*}
&\{H_j,H_c\}(\textbf{x}^0)=\{H_j,H_c\}(\textbf{x}(t))=\lim_{t\rightarrow -\infty}\{H_j,H_c\}(\textbf{x}(t))\\
=&\lim_{t\rightarrow -\infty}\left\{\sum_{J \in \binom{[K]}{j}}\prod_{p\in J}m_{2p-1}m_{2p}e^{x_{2p-1}-x_{2p}},\prod_{q=1}^{2K+1}m_qe^{(-1)^{q+1}x_q}\right\}(\textbf{x}(t))\\
=&\lim_{t\rightarrow -\infty}\sum_{J \in \binom{[K]}{j}}\left\{\prod_{p\in J}m_{2p-1}m_{2p}e^{x_{2p-1}-x_{2p}},\prod_{q=1}^{2K+1}m_qe^{(-1)^{q+1}x_q}\right\}(\textbf{x}(t))\\
=&\lim_{t\rightarrow -\infty}\sum_{J \in \binom{[K]}{j}}\prod_{p\in J}m_{2p-1}m_{2p}e^{x_{2p-1}-x_{2p}}\prod_{q=1}^{2K+1}m_qe^{(-1)^{q+1}x_q}\\
&\qquad\qquad\qquad\cdot\sum_{p\in J,}\left(\sum_{q=1}^K\left\{x_{2p-1}-x_{2p},x_{2q-1}-x_{2q}\right\}+\left\{x_{2p-1}-x_{2p},x_{2K+1}\right\}\right)(\textbf{x}(t))
\end{align*}
We have shown in the course of proving \autoref{lem:lem1} that
$$\left\{x_{2p-1}-x_{2p},x_{2q-1}-x_{2q}\right\}=0.$$
Furthermore, 
$$\left\{x_{2p-1}-x_{2p},x_{2K+1}\right\}=\sgn(x_{2p-1}-x_{2K+1})-\sgn(x_{2p}-x_{2K+1})=0$$
holds.  
Thus, indeed
$$\{H_j,H_c\}(\textbf{x}^0)=0,$$
and the proof is completed.
\end{proof}
\end{enumerate}

\begin{theorem} \label{thm:Liouville integrability} 
The conservative peakon system given by \autoref{eq:peakonODEs} 
is Liouville integrable.  
\end{theorem} 
\begin{proof} 
To prove Liouville integrability of a Hamiltonian system 
defined on an $n$-dimensional  Poisson manifold $(M, \pi)$ we need to show that 
the integrals of motion are functionally independent, they commute and 
the number of them, say $s$,  satisfies $\frac12 rank(\pi)+s=n$  \cite{poisson}.  In the case at hand, by \autoref{lem:rank pi}, $rank(\pi)=K$ and since we have $s=K$ commuting Hamiltonians for $n=2K$ and 
$s=K+1$ commuting  Hamiltonians for $n=2K+1$,  the relation 
$\frac12 rank(\pi)+s=n$ holds in both cases.  
To prove the functional independence we proceed as follows.  
We need to prove that $\Omega=dH_1\wedge dH_2\wedge\dots \wedge dH_s$ is non-zero on a dense subset of $M$.  
Since $H_i$s are rational 
functions of $\xi_l=e^{x_l}, l=1,\dots, n$, so is $\Omega$.  It 
is thus sufficient to show that $\Omega$ is not identically zero.  
In the asymptotic regions used in the proofs of \autoref{thm:even commute}
and \autoref{thm:odd commute} $H_j$s are asymptotically polynomials in the variables 
$e^{x_{2p-1}-x_{2p}}$ of degree $j$ (see e.g. \eqref{eq:Hi-asympt}) 
and, conseqently, $\Omega$ is not identically zero in that region, which 
concludes the proof.  
\end{proof} 

We finish this section by observing that it is elementary to 
identify the Darboux coordinates in the case $n=2K$.  Indeed, 
suppose we introduce coordinates $$\{p_k,\ q_k,\ \ k=1,2,...,K\},$$
  by defining 
  $$p_k=x_{2k}-x_{2k-1},\ \ \ q_k=\sum_{j=1}^k(x_{2j-1}-x_{2j-2})$$
   with the convention that $x_0=0$. Then by a direct computation we obtain
    $$\{p_k,q_m\}=\delta_{km},\ \ \  \{p_k,p_m\}=0,\ \ \  \{q_k,q_m\}=0.$$
Thus $\{p_k,\ q_k, \ k=1,2,...,K\}$ can be regarded as Darboux coordinates.
\section{Concluding Remarks} 
The peakon sector of the modified Camassa-Holm equation 
$m_t+((u_x^2-u^2)m)_x=0$ is not uniquely defined.  
One way of defining it proposed in \cite{gui2013wave, qiao-xia}
has a feature that even though the Sobolev $H^1$ norm 
$||u||_{H^1} ^2$ is one of the Hamiltonians of the bi-Hamiltonian 
formulation of this equation in the smooth sector its peakon counterpart 
does not have this property.  In \cite{chang-szmigielski-short1mCH} we 
proposed a different definition of mCH peakons, based on a different 
regularization of the ill-defined term $u_x^2 m$.  For these (conservative) 
peakons the $H^1$ norm is preserved .  In this paper we 
amplify this statement in the following way: 
\begin{enumerate} 
\item we introduce a Poisson structure relative to which 
the conservative peakon equations are Hamiltonian with the 
Hamiltonian being the very norm $||u||^2_{H^1}$; 
\item using the inverse spectral solutions to conservative 
peakons put forward in 
\cite{chang-szmigielski-m1CHlong} we show the Liouville integrability of 
the conservative peakon system. 
\end{enumerate} 
We conclude this Letter by emphasizing that \autoref{eq:m1CH} is the 
first equation known to us which has two, natural, peakon sectors.  
It remains an open question if there are other peakon equations 
exhibiting this property and, ultimately, what purpose, mathematical or 
physical, the very existence of these two types of singular solitons imparts 
to the subject. 
\section{Acknowledgements} 
The authors thank Stephen Anco for his interest in this work.  
The Poisson structure used in this Letter has been independently 
derived by him and announced in his talk during the conference `` 
Nonlinear Evolution Equations and Wave Phenomena'', Athens, GA, March 2017.  The authors thank Yue Liu for an interesting discussion 
regarding the dissipative peakons.  The authors also thank the 
referee for critical remarks and for bringing references 
\cite{estevez, estevez-sardon} to their attention. Our comments in the 
introduction on the role 
of the reciprocal transformation are partially in response to his/her 
suggestions.

Xiangke Chang was supported in part by the Natural  Sciences and Engineering Research Council of Canada (NSERC), the Department of Mathematics and Statistics of the University of Saskatchewan, PIMS postdoctoral fellowship and 
the Institute of Computational Mathematics, AMSS, CAS.  Jacek Szmigielski was supported in part by NSERC \#163953.

\bibliographystyle{abbrv}
\def\cydot{\leavevmode\raise.4ex\hbox{.}}
  \def\cydot{\leavevmode\raise.4ex\hbox{.}}

\end{document}